\documentclass[a4paper]{article}

\usepackage{graphicx}
\usepackage{color}
\usepackage{stmaryrd}
\usepackage[all]{xy}

\usepackage[numbers]{natbib}
\usepackage{natbibspacing}
\renewcommand{\bibfont}{\small}

\definecolor{myurlcolor}{rgb}{0.1,0.1,0.8}
\definecolor{mylinkcolor}{rgb}{0.05,0.05,0.4}
\usepackage{hyperref}
\hypersetup{colorlinks,
linkcolor=mylinkcolor,
citecolor=mylinkcolor,
urlcolor=mylinkcolor}

\usepackage{tocloft}
\usepackage{latexsym}
\usepackage{amssymb,amsmath}
\usepackage{mathrsfs}
\usepackage{color}

\newcommand{\bref}[1]{(\ref{#1})}

\newcommand{\twid}[1]{\widetilde{#1}}

\newcommand{\such}{:}

\def\today{\number\day\space \ifcase\month\or
  January\or February\or March\or April\or May\or June\or
  July\or August\or September\or October\or November\or December\fi
  \space\number\year}

\usepackage{marginnote}

\usepackage{stmaryrd}

\usepackage{graphicx}

\newcommand{\reals}{\mathbb{R}}

{\begin{array}{c}
\setlength{\unitlength}{1em}}%
{\end{array}}
{\begin{array}{l}}%
{\end{array}}
{\begin{array}{c}}%
{\end{array}}

\newcommand{\demph}[1]{\textbf{\textup{#1}}}
\newcommand{\of}{\circ}

\definecolor{grey}{gray}{0.5}

\newcommand{\slogan}[1]%
{\begin{center}\textit{#1}\end{center}}

\newcommand{\dee}{\,d}

\newcommand{\oppair}[4]{%
\xymatrix@1{%
#1 \ar@<.5ex>[r]^{#3} &{#2}\ar@<.5ex>[l]^{#4}%
}}

\newcommand{\from}{\colon}

\newcommand{\hadjnli}[4]{%
\xymatrix@1{
#1 \ar@<1ex>[r]^-{#3} \ar@{}[r]|-\bot &#2 \ar@<1ex>[l]^-{#4}}}
\newcommand{\hadjnri}[4]{%
\xymatrix@1{
#1 \ar@<1ex>[r]^-{#3} \ar@{}[r]|-\top &#2 \ar@<1ex>[l]^-{#4}}}

{\hspace*{-.5em}%
\end{array}}

{\end{array}%
}

{\end{array}%
}

\newcommand{\twocelli}[5]{%
\xymatrix@1{
#1 \ar@/^.8pc/[r]^-{#3} 
\ar@/_.8pc/[r]_-{#4} 
\ar@{}[r]|-{\,{\displaystyle\Downarrow}\,#5} &#2}}

\newcommand{\R}{\reals}

\newcommand{\vc}[1]{\mathbf{#1}} % Vector

\newcommand{\mjpar}[1]{}
\newcommand{\mapar}[1]{}
\newcommand{\mmpar}[1]{}

\newcommand{\p}{\vc{p}}
\newcommand{\w}{\vc{w}}
\newcommand{\relent}[2]{H(#1 \mathbin{\|} #2)}
\newcommand{\Relent}[2]{H\bigl(#1 \mathbin{\big\|} #2\bigr)}
\newcommand{\melent}[2]{I(#1 \mathbin{\|} #2)}
\newcommand{\Melent}[2]{I\bigl(#1 \mathbin{\big\|} #2\bigr)}
\newcommand{\MelEnt}[2]{I\Bigl(#1 \mathbin{\Big\|} #2\Bigr)}

\newcommand{\srelent}[3]{S_{#1}(#2 \mathbin{\|} #3)}
\newcommand{\Srelent}[3]{S_{#1}\bigl(#2 \mathbin{\big\|} #3\bigr)}

\newcommand{\rhow}{R}

\newcommand{\csuch}{\colon}

\usepackage[amsmath,thmmarks]{ntheorem}

\newtheorem{thm}{Theorem}[section]

\newtheorem{lemma}[thm]{Lemma}

\theorembodyfont{\normalfont}

\newtheorem{remark}[thm]{Remark}
\newtheorem{remarks}[thm]{Remarks}

\theoremstyle{nonumberplain}
\theoremsymbol{\ensuremath{\square}}
\qedsymbol{\ensuremath{\square}}

\newtheorem{proof}{Proof}

\newcommand{\theoremtobeproved}{}
\newtheorem{pfoftheorem}{Proof of \theoremtobeproved}
\newenvironment{pfof}[1]
{
\renewcommand{\theoremtobeproved}{#1}
\begin{pfoftheorem}
}
{\end{pfoftheorem}}

\title{A short characterization of relative entropy}
\author{Tom Leinster%
\thanks{School of Mathematics, University of Edinburgh, UK;
  Tom.Leinster@ed.ac.uk.
}}
\date{}

\begin{document}

\sloppy
\maketitle

\begin{abstract}\noindent
We prove characterization theorems for relative entropy (also known as
Kullback--Leibler divergence), $q$-logarithmic entropy (also known as
Tsallis entropy), and $q$-logarithmic relative entropy.  All three
have been characterized axiomatically before, but we show that earlier
proofs can be simplified considerably, at the same time relaxing some of
the hypotheses. 
\end{abstract}

\section{Introduction}

The Shannon entropy of a finite probability distribution $\p = (p_1,
\ldots, p_n)$, 
\[
H(\p)
=
\sum_{i \csuch p_i > 0} p_i \log \frac{1}{p_i},
\]
is such an important quantity that many authors have sought short lists of
properties that determine $H$ uniquely.  Many such characterization
theorems have been found, beginning with one in Shannon's seminal paper of
1948 (\cite{ShanMTC}, Theorem~2).  For instance, Faddeev~\cite{Fadd} proved
that up to a constant factor, $H$ is uniquely characterized by symmetry,
continuity, and a certain recursivity property.

Accompanying Shannon entropy is the concept of relative
entropy, defined as follows.  Given probability distributions $\p$ and
$\vc{r}$ on $n$ elements, the \demph{entropy of $\p$ relative to $\vc{r}$}
is
\[
\relent{\p}{\vc{r}}
=
\sum_{i \csuch p_i > 0} p_i \log \frac{p_i}{r_i}
\in [0, \infty].
\]
Relative entropy goes by a multitude of names: Kullback--Leibler
divergence, directed divergence, discrimination information, relative
information, information gain, and so on.  In information theory, it
measures the wastage when a language whose $n$ letters have frequencies $\p
= (p_1, \ldots, p_n)$ is encoded using a system optimized for a different
language with frequencies $\vc{r}$, instead of the system optimized for the
original language.
% (See Cover and
% Thomas~\cite{CoTh2}, Exercise~5.30).  
There are other interpretations in other fields, as the plethora of
names suggests.

Axiomatic characterizations of relative entropy have also been sought and
found.  One such theorem is implicit in work of Kannappan and
Ng~\cite{KaNgMSF}.  It states that up to a constant factor, relative
entropy is uniquely determined by measurability in each of $\p$ and
$\vc{r}$ separately, invariance under permutations of $\{1, \ldots, n\}$,
the vanishing property $\relent{\p}{\p} = 0$, and a certain recursivity
equation.  (Remark~\ref{rmk:relent-history} gives further details.)
Their proof was a tour de force of functional
equations, involving the solution of the functional equation
\begin{equation}
\label{eq:feith-gen}
f(x) + (1 - x) g \biggl( \frac{y}{1 - x} \biggr)
=
h(y) + (1 - y) k \biggl( \frac{x}{1 - y} \biggr)
\end{equation}
in four unknown functions, as well as the four-variable functional equation
\[
F(x, y) + (1 - x) 
F \biggl( \frac{u}{1 - x}, \frac{v}{1 - y} \biggr)
=
F(u, v) + (1 - u)
F \biggl( \frac{x}{1 - u}, \frac{y}{1 - v} \biggr).
\]
We give a much simpler proof, at the same time weakening the measurability
hypothesis.  Our proof involves neither of these equations.  Instead, it
borrows heavily from a categorical characterization of relative entropy by
Baez and Fritz~\cite{BaFr}.  This is the first main result,
Theorem~\ref{thm:relent}.

Shannon entropy is just one member (albeit a special one) of a
one-parameter family of entropies $(S_q)_{q \in \R}$, first investigated by
Havrda and Charv\'at~\cite{HaCh} and often misattributed to Tsallis
(Remark~\ref{rmks:q-history}\bref{rmk:q-history-ent}).  These entropies
$S_q$, and the accompanying relative entropies, are defined as follows.

For $q \in \R$, the \demph{$q$-logarithm} is the function $\ln_q \from (0,
\infty) \to \R$ given by
\[
\ln_q(x) = \int_1^x t^{-q} \dee t.
\]
The \demph{$q$-logarithmic entropy} and \demph{$q$-logarithmic
  relative entropy} are defined by
\begin{align*}
S_q(\p) &
=
\sum_{i \csuch p_i > 0} p_i \ln_q \frac{1}{p_i},        \\
\srelent{q}{\p}{\vc{r}} &
=
-
\sum_{i \csuch p_i > 0} p_i \ln_q \frac{r_i}{p_i},
\end{align*}
for probability distributions $\p$ and $\vc{r}$ on $n$ elements.  When $q =
1$, these reduce to the ordinary Shannon entropy and relative entropy.  
% They can be
% regarded as deformations of ordinary Shannon entropy and relative entropy.

There are several existing theorems characterizing the $q$-logarithmic
entropy for a given $q \neq 1$.  Up until now, the simplest appears to have
been the 1970 result of Dar\'oczy~\cite{DaroGIF}.  We simplify further,
weakening the hypotheses and shortening the proof to just a few lines
(Theorem~\ref{thm:q}).

Finally, we use a similar and equally short argument to characterize the
$q$-logarithmic relative entropies $\srelent{q}{-}{-}$
(Theorem~\ref{thm:q-rel}).

It is remarkable that
% these last two characterization theorems involve no 
when $q \neq 1$, the characterizations of
$q$-logarithmic entropy and $q$-logarithmic relative entropy need no
regularity conditions whatsoever (not even measurability), in contrast to
the theorems for $q = 1$.

The remaining three sections of this paper establish our three theorems in
turn, characterizing first relative entropy, then $q$-logarithmic entropy,
then $q$-logarithmic relative entropy.

\section{Relative entropy}

For $n \geq 1$, write
\[
\Delta_n 
=
\Bigl\{ 
\p = (p_1, \ldots, p_n) \in \R^n
\such 
p_i \geq 0, \ \sum p_i = 1
\Bigr\}
\]
for the set of probability distributions on $\{1, \ldots, n\}$, and write
\[
A_n 
=
\bigl\{
(\p, \vc{r}) \in \Delta_n \times \Delta_n
\such
p_i = 0 \text{ whenever } r_i = 0
\bigr\}.
\]
Evidently, $(\p, \vc{r}) \in A_n$ if
and only if the relative entropy
\[
\relent{\p}{\vc{r}}
=
\sum_{i \csuch p_i > 0} p_i \log \frac{p_i}{r_i}
\]
is finite.  (Viewing $\p$ and $\vc{r}$ as measures on $\{1, \ldots, n\}$,
we have $(\p, \vc{r}) \in A_n$ just when $\p$ is absolutely continuous with
respect to $\vc{r}$.)  

We will characterize the sequence of functions
\begin{equation*}
% \label{eq:relent-fns}
\bigl(
\relent{-}{-} \from A_n \to \R
\bigr)_{n \geq 1}
\end{equation*}
uniquely up to a constant factor.
It is easy to check that this sequence 
% of functions~\eqref{eq:relent-fns}
has the following four properties, as does any scalar multiple
$c\relent{-}{-}$ (for $c \in \R$).

\begin{description}
\item[Measurability in the second argument] 
For each $n \geq 1$ and $\vc{p} \in \Delta_n$, the function
\[
\begin{array}{ccc}
\{ \vc{r} \in \Delta_n \such (\p, \vc{r}) \in A_n \}        &
\to     &
\R      \\
\vc{r}  &
\mapsto &
\relent{\p}{\vc{r}}
\end{array}
\]
is Lebesgue measurable.

\item[Symmetry]
For each $n \geq 1$, $(\p, \vc{r}) \in A_n$ and permutation $\sigma$ of
$\{1, \ldots, n\}$, 
\begin{equation}
\label{eq:relent-sym}
\relent{\p}{\vc{r}}
=
\relent{\p\sigma}{\vc{r}\sigma},
\end{equation}
where $\p\sigma = (p_{\sigma(1)}, \ldots, p_{\sigma(n)})$.

\item[Vanishing] 
$\relent{\p}{\p} = 0$ for all $n \geq 1$ and $\p \in \Delta_n$.

\item[Chain rule]
To state this, we need some notation.  Given $n, k_1, \ldots, k_n \geq 1$ and
$\vc{w} \in \Delta_n, \p^1 \in \Delta_{k_1}, \ldots, \p^n \in
\Delta_{k_n}$, and writing $\p^i = (p^i_1, \ldots, p^i_{k_i})$, define 
\begin{align*}
% \[
\vc{w} \of (\p^1, \ldots, \p^n)
&
=
(w_1 p^1_1, \ldots, w_1 p^1_{k_1},
\ \ldots, \ 
w_n p^n_1, \ldots, w_n p^n_{k_n})       % \\
% &
\in
\Delta_{k_1 + \cdots + k_n}.
% \]
\end{align*}
The \demph{chain rule} for relative entropy is that 
\begin{equation}
\label{eq:relent-chn}
\Relent{\vc{w} \of (\p^1, \ldots, \p^n)}%
{\twid{\vc{w}} \of (\twid{\p}^1, \ldots, \twid{\p}^n)}  
=
\relent{\vc{w}}{\twid{\vc{w}}} 
+ \sum_{i = 1}^n w_i \relent{\p^i}{\twid{\p}^i} 
\end{equation}
whenever $(\vc{w}, \twid{\vc{w}}) \in A_n$ and $(\p^i, \twid{\p}^i) \in
A_{k_i}$.  
(Under these hypotheses, the pair of
distributions on the left-hand side 
% of equation~\eqref{eq:relent-chn}
belongs to $A_{k_1 + \cdots + k_n}$.)
\end{description}

\begin{thm}
\label{thm:relent}
Let $\bigl( \melent{-}{-} \from A_n \to \R \bigr)_{n \geq 1}$ be a sequence
of functions.  The following are equivalent:
\begin{enumerate}
\item 
\label{part:relent-condns}
$\melent{-}{-}$ satisfies the four properties above: measurability in the
  second argument, symmetry, vanishing, and the chain rule;

\item
\label{part:relent-form}
$\melent{-}{-} = c\relent{-}{-}$ for some $c \in \R$.
\end{enumerate}
\end{thm}

We have just noted that~\bref{part:relent-form}
implies~\bref{part:relent-condns}.  We now embark on the proof of the
converse.  
For the rest of this section, let $\bigl( \melent{-}{-} \from
A_n \to \R \bigr)_{n \geq 1}$ be a sequence of functions satisfying the
four conditions.
Define a function $L \from (0, 1] \to \R$ by
\[
L(\alpha) = \Melent{(1, 0)}{(\alpha, 1 - \alpha)}.
\]
The idea is that if $\melent{-}{-} = \relent{-}{-}$ then $L = -\log$.  We
will show that in any case, $L$ is a scalar multiple of $\log$.

\begin{lemma}
\label{lemma:zeros}
Let $(\p, \vc{r}) \in A_n$ with $p_{k + 1} = \cdots = p_n = 0$, where $1
\leq k \leq n$.  Then $r_1 + \cdots + r_k > 0$ and
\[
\melent{\p}{\vc{r}} 
=
L(r_1 + \cdots + r_k) 
+
\melent{\p'}{\vc{r}'},
\]
where 
\[
\p' = (p_1, \ldots, p_k),
\qquad
\vc{r}' = \frac{(r_1, \ldots, r_k)}{r_1 + \cdots + r_k}.
\]
\end{lemma}

\begin{proof}
The case $k = n$ reduces to the statement that $L(1) = 0$, which follows
from the vanishing property.  Suppose, then, that $k < n$.

Since $\p$ is a probability distribution with $p_i = 0$ for all $i > k$,
there is some $i \leq k$ such that $p_i > 0$, and then $r_i > 0$ as
$(\p, \vc{r}) \in A_n$.  Hence $r_1 + \cdots + r_k > 0$.  Let $\vc{r}'' \in
\Delta_{n - k}$ be the normalization of $(r_{k + 1}, \ldots, r_n)$ if $r_{k
  + 1} + \cdots + r_n > 0$, or choose $\vc{r}''$ arbitrarily in $\Delta_{n
  - k}$ otherwise (which is possible since $k < n$).  Then 
\[
\melent{\p}{\vc{r}}     
=
\Melent{(1, 0) \of (\p', \vc{r}'')}
{(r_1 + \cdots + r_k, \, r_{k + 1} + \cdots + r_n) \of 
(\vc{r}', \vc{r}'')}.
\]
The result now follows from the chain rule.
\end{proof}

\begin{lemma}
\label{lemma:two-add}
$L(\alpha\beta) = L(\alpha) + L(\beta)$ for all $\alpha, \beta \in (0,
1]$. 
\end{lemma}

\begin{proof}
We evaluate the real number
\[
x := 
\Melent{(1, 0, 0)}{(\alpha\beta, \alpha(1 - \beta), 1 - \alpha)}
\]
in two ways.  By Lemma~\ref{lemma:zeros} with $k = 1$ and
the vanishing property,
\[
x 
= 
L(\alpha\beta) + \Melent{(1)}{(1)}
=
L(\alpha\beta),
\]
where $(1)$ is the unique element of $\Delta_1$.  But also, by
Lemma~\ref{lemma:zeros} with $k = 2$,
\[
x 
=
L(\alpha) + \Melent{(1, 0)}{(\beta, 1 - \beta)}
=
L(\alpha) + L(\beta).
\]
Comparing the two expressions for $x$ gives the result.
% 
% By Lemma~\ref{lemma:zeros} with $k = 1$ and
% the vanishing property,
% \[
% \Melent{(1, 0, 0)}{(\alpha\beta, \alpha(1 - \beta), 1 - \alpha)}
% =
% L(\alpha\beta) + \Melent{(1)}{(1)}
% =
% L(\alpha\beta),
% \]
% where $(1)$ is the unique element of $\Delta_1$.  But also, by
% Lemma~\ref{lemma:zeros} with $k = 2$,
% \[
% \Melent{(1, 0, 0)}{(\alpha\beta, \alpha(1 - \beta), 1 - \alpha)}
% =
% L(\alpha) + \Melent{(1, 0)}{(\beta, 1 - \beta)}
% =
% L(\alpha) + L(\beta).
% \]
% Comparing the two equations gives the result.
% 
\end{proof}

\begin{lemma}
\label{lemma:rel-two-log}
There is some $c \in \R$ such that $L(\alpha) = -c\log\alpha$ for all
$\alpha \in (0, 1]$.
\end{lemma}

\begin{proof}
Define $f \from [0, \infty) \to \R$ by $f(t) = L(e^{-t})$.  By
Lemma~\ref{lemma:two-add}, $f(t + u) = f(t) + f(u)$ for all $t, u \in [0,
\infty)$.  Also, $f$ is measurable, since $L$ is.  It is
well-known~\cite{BanaSEF} that these conditions force $f(t) = ct$ for
some constant $c$, giving $L(\alpha) = -c\log\alpha$.
\end{proof}

Our next lemma is an adaptation of the most ingenious part of Baez and
Fritz's argument (\cite{BaFr}, Lemma~4.2).

\begin{lemma}
\label{lemma:bf-full-supp}
Let $(\p, \vc{r}) \in A_n$ with $p_i > 0$ for all $i \in \{1, \ldots, n\}$.
Then $\melent{\p}{\vc{r}} = c\relent{\p}{\vc{r}}$.
\end{lemma}

\begin{proof}
The hypotheses imply that $r_i > 0$ for all $i$.  We can
therefore choose some $\alpha \in (0, 1]$ such that $r_i - \alpha p_i \geq 0$
  for all $i$.
We will compute the number
\[
x
:=
\Melent{(p_1, \ldots, p_n, \underbrace{0, \ldots, 0}_n)}
{(\alpha p_1, \ldots, \alpha p_n, 
r_1 - \alpha p_1, \ldots, r_n - \alpha p_n)} 
\]
in two ways.  
% (The pair of distributions on the right-hand side belongs to
% $A_{2n}$, so $x$ is well-defined.)  
First, by Lemma~\ref{lemma:zeros}, Lemma~\ref{lemma:rel-two-log},
and the vanishing property,
\[
x 
=
L(\alpha) + \melent{\p}{\p}
=
-c\log\alpha.
\]
Second, by symmetry, the chain rule, and Lemma~\ref{lemma:rel-two-log},
\begin{align*}
x       &
=
\Melent{(p_1, 0, \ldots, p_n, 0)}
{(\alpha p_1, r_1 - \alpha p_1, \ldots, p_n, r_n - \alpha p_n)} \\
&
=
\MelEnt{\p \of \bigl((1, 0), \ldots, (1, 0)\bigr)}
{\vc{r} \of 
\bigl(
\bigl(\alpha \tfrac{p_1}{r_1}, 1 - \alpha \tfrac{p_1}{r_1}\bigr), \ldots, 
\bigl(\alpha \tfrac{p_n}{r_n}, 1 - \alpha \tfrac{p_n}{r_n}\bigr)
\bigr)} \\
&
=
\melent{\p}{\vc{r}} + 
\sum_{i = 1}^n p_i L\bigl(\alpha \tfrac{p_i}{r_i}\bigr) \\
&
=
\melent{\p}{\vc{r}} - c\log \alpha - c\relent{\p}{\vc{r}}.
\end{align*}
Comparing the two expressions for $x$ gives the result.
\end{proof}

We have now proved that $\melent{\p}{\vc{r}} = c\relent{\p}{\vc{r}}$ when
$\p$ has full support.  It only remains to prove it for arbitrary $\p$. 

\begin{pfof}{Theorem~\ref{thm:relent}}
Let $(\p, \vc{r}) \in A_n$.  By symmetry, we can assume that $p_1, \ldots,
p_k > 0$ and $p_{k + 1} = \cdots = p_n = 0$ for some $k \in \{1, \ldots,
n\}$.  Writing $\rhow = r_1 + \cdots + r_k$, we have $R > 0$ since $(\p,
\vc{r}) \in A_n$, and t
\[
\melent{\p}{\vc{r}}     
=
L(\rhow) 
+ 
\Melent{(p_1, \ldots, p_k)}{\tfrac{1}{\rhow}(r_1, \ldots, r_k)}
\]
by Lemma~\ref{lemma:zeros}.  Hence by Lemmas~\ref{lemma:rel-two-log}
and~\ref{lemma:bf-full-supp},
\[
\melent{\p}{\vc{r}}     
=
-c\log R
+ c\Relent{(p_1, \ldots, p_k)}{\tfrac{1}{\rhow}(r_1, \ldots, r_k)}.
\]
But by the same argument applied to $cH$ in place of $I$ (or by direct
calculation), we also have
\[
c\relent{\p}{\vc{r}}     
=
-c\log R
+ c\Relent{(p_1, \ldots, p_k)}{\tfrac{1}{\rhow}(r_1, \ldots, r_k)}.
\]
The result follows.
\end{pfof}

\begin{remarks}
\label{rmks:relent-condns}
\begin{enumerate}
\item 
\label{rmk:relent-condns-van}
The vanishing axiom cannot be dropped from Theorem~\ref{thm:relent}.
Indeed, the quantity
$
\sum_{i \csuch p_i > 0} p_i \log \frac{1}{r_i}
$
% (often called the \demph{cross entropy}~\cite{GoodSTN}) 
satisfies the other
three axioms but not vanishing.

\item
\label{rmk:relent-condns-chn}
In the literature on information functions, the chain rule is often
replaced by one of two superficially simpler rules.  The first is the
special case $k_1 = 2, k_2 = \cdots = k_n = 1$, which is
\begin{multline}
\label{eq:recursivity}
\Relent{(pw_1, (1 - p)w_1, w_2, \ldots, w_n)}%
{(\twid{p}\twid{w}_1, (1 - \twid{p})\twid{w}_1, \twid{w}_2, \ldots,
  \twid{w}_n)}
\\
=
\Relent{\vc{w}}{\twid{\vc{w}}}
+
w_1 \Relent{(p, 1 - p)}{(\twid{p}, 1 - \twid{p})}
\end{multline}
($(\vc{w}, \twid{\vc{w}}) \in A_n$, $((p, 1 - p), (\twid{p}, 1 - \twid{p}))
\in A_2$).  This is known as \demph{recursivity}.  The second is the
special case $n = 2$ of the chain rule, which is
\begin{multline}
\label{eq:other-special-case}
\Relent{w\p \oplus (1 - w)\vc{r}}%
{\twid{w}\twid{\p} \oplus (1 - \twid{w})\twid{\vc{r}}}
\\
=
\Relent{(w, 1 - w)}{(\twid{w}, 1 - \twid{w})}
+ w \relent{\p}{\twid{\p}} 
+ (1 - w) \relent{\vc{r}}{\twid{\vc{r}}},
\end{multline}
where
\[
w\p \oplus (1 - w)\vc{r}
=
(wp_1, \ldots, wp_k, (1 - w)r_1, \ldots, (1 - w)r_\ell)
\]
and $( (w, 1 - w), (\twid{w}, 1 - \twid{w}) ) \in A_2$, $(\p, \twid{\p})
\in A_k$, $(\vc{r}, \twid{\vc{r}}) \in A_\ell$.  However, straightforward
inductions (similar to those in Feinstein~\cite{Fein}, p.~5--6) show that
in the presence of the symmetry axiom, either one of the special
cases~\eqref{eq:recursivity} or~\eqref{eq:other-special-case} is
equivalent to the full chain rule~\eqref{eq:relent-chn}.  Which to use is,
therefore, simply a matter of taste.
\end{enumerate}
\end{remarks}

\begin{remark}
\label{rmk:relent-history}
Here we compare Theorem~\ref{thm:relent} with some earlier
characterizations of relative entropy.  One of the first such theorems was
that of Hobson~\cite{Hobs}, who used stronger hypotheses for the same
conclusion.  In common with Theorem~\ref{thm:relent}, he assumed symmetry,
vanishing, and the chain rule (in the equivalent
form~\eqref{eq:other-special-case}).  But he also assumed continuity in
both variables (instead of measurability in one) and a
monotonicity hypothesis unlike anything in Theorem~\ref{thm:relent}.

In 1973, Kannappan and Ng~\cite{KaNgMSF} proved a result very close to
Theorem~\ref{thm:relent}.  They did not \emph{state} that result
in~\cite{KaNgMSF}, but the closing remarks in another paper by the same
authors~\cite{KaNgFEC} and the approach of a contemporaneous paper by
Kannappan and Rathie~\cite{KaRa} strongly suggest the intent.  The
result was stated explicitly by
Csisz\'ar (\cite{Csis}, Section~2.1), who attributed it to
Kannappan and Ng.

There are some slight differences of hypotheses between Kannappan and Ng's
theorem and Theorem~\ref{thm:relent}.  They assumed measurability in both
variables, whereas we only assumed measurability in the second.  (In fact,
all we used was that $\melent{(1, 0)}{-}$ is measurable.)  On the other
hand, they only needed the vanishing condition for $(1/2, 1/2)$, whereas we
needed it for all $\p$.  They used the chain rule in the equivalent
form~\eqref{eq:recursivity}.  And as indicated in the Introduction, the
proofs are entirely different.
\end{remark}

\section{$q$-logarithmic entropy}

Let $q \in \R$.  The definition of $q$-logarithm in the Introduction
gives, explicitly,  
\[
\ln_q(x) 
=
\frac{1}{1 - q} (x^{1 - q} - 1)
\]
for $x \in (0, \infty)$ and $q \neq 1$, while $\ln_1$ is the natural
logarithm $\log$.  Hence, explicitly, the $q$-logarithmic entropy is given
by
\[
S_q(\p) 
=
\frac{1}{1 - q} \Biggl( \sum_{i \csuch p_i > 0} p_i^q - 1 \Biggr)
\]
for $\p \in \Delta_n$ and $q \neq 1$, while $S_1$ is the Shannon entropy
$H$.  We have $\ln_q(x) \to \log(x)$ as $q \to 1$, hence also $S_q(\p) \to
H(\p)$ as $q \to 1$.

Fix $q \in \R$. The $q$-logarithmic entropy satisfies a chain rule
\begin{equation}
\label{eq:q-chn}
S_q\bigl( \vc{w} \of (\p^1, \ldots, \p^n) \bigr)
=
S_q(\vc{w}) + \sum_{i \csuch w_i > 0} w_i^q S_q(\p^i)
\end{equation}
($\w \in \Delta_n$, $\p^1 \in \Delta_{k_1}, \ldots, \p^n \in
\Delta_{k_n}$), as is easily checked.  In particular, this holds when $\p^1
= \cdots = \p^n = \p$, say.  For $\w \in \Delta_n$ and $\p \in \Delta_k$,
write
\begin{align*}
\vc{w} \otimes \p       &
=
\vc{w} \of (\p, \ldots, \p)     \\
&
=
(w_1 p_1, \ldots, w_1 p_k,
\ \ldots, \ 
w_n p_1, \ldots, w_n p_k)       
\in \Delta_{nk}.
\end{align*}
In this case, the $q$-chain rule~\eqref{eq:q-chn} gives a
\demph{$q$-multiplicativity} property:
\begin{equation}
\label{eq:q-mult}
S_q(\vc{w} \otimes \p)
=
S_q(\vc{w}) + 
\Biggl( \sum_{i \csuch w_i > 0} w_i^q \Biggr) S_q(\p)
\end{equation}
($n, k \geq 1$, $\vc{w} \in \Delta_n$, $\p \in \Delta_k$).  

Note also that $S_q$ is symmetric in its arguments:
\begin{equation}
\label{eq:q-sym}
S_q(\p)
=
S_q(\p\sigma)
\end{equation}
for all $\p \in \Delta_n$ and permutations $\sigma$ of
$\{1, \ldots, n\}$.  

% Evidently, both~\eqref{eq:q-mult} and~\eqref{eq:q-sym}
% also hold for any constant multiple of $S_q$.

The left-hand side of equation~\eqref{eq:q-mult} is symmetric in $\vc{w}$
and $\p$, but the right-hand side is not obviously so.  This is the key to
our second theorem.

\begin{thm}
\label{thm:q}
Let $1 \neq q \in \R$ and let $(I \from \Delta_n \to \R)_{n \geq 1}$ be a
sequence of functions.  The following are equivalent:
\begin{enumerate}
\item 
\label{part:q-condns}
$I$ has the $q$-multiplicativity property~\eqref{eq:q-mult} and the
symmetry property~\eqref{eq:q-sym} (both with $I$ in place of $S_q$);

\item
\label{part:q-form}
$I = cS_q$ for some $c \in \R$.
\end{enumerate}
\end{thm}

\begin{proof}
By the observations just made, \bref{part:q-form}
implies~\bref{part:q-condns}.  Now assume~\bref{part:q-condns}.  By
symmetry, $I(\vc{w} \otimes \vc{p}) = I(\vc{p} \otimes \vc{w})$, so
\[
I(\vc{w}) + \Biggl( \sum_{i \csuch w_i > 0} w_i^q \Biggr) I(\vc{p})
=
I(\vc{p}) + \Biggl( \sum_{i \csuch p_i > 0} p_i^q \Biggr) I(\vc{w}),
\]
or equivalently
\[
\Biggl( \sum_{i \csuch w_i > 0} w_i^q - 1 \Biggr) I(\vc{p})
=
\Biggl( \sum_{i \csuch p_i > 0} p_i^q - 1 \Biggr) I(\vc{w}),
\]
for all $\vc{w} \in \Delta_n$ and $\vc{p} \in \Delta_k$.  Take $\vc{w} =
(1/2, 1/2)$: then for all $\p \in \Delta_k$,
\[
\bigl(2^{1 - q} - 1\bigr) I(\p)
=
\Biggl( \sum_{i \csuch p_i > 0} p_i^q - 1 \Biggr) I(1/2, 1/2).
\]
Since $q \neq 1$, we can define $c =
\tfrac{1 - q}{2^{1 - q} - 1} \cdot I(1/2, 1/2)$, and then $I = cS_q$.
\end{proof}

\begin{remarks}
\label{rmks:q-history}
\begin{enumerate}
\item
\label{rmk:q-history-log}
The $q$-logarithms were used in Hardy, Littlewood and P\'olya's classic
book on inequalities, first published in 1934 (\cite{HLP}, proof of
Theorem~84).  They have been an explicit object of study since at
least a 1964 paper of Box and Cox in statistics (\cite{BoCo}, Section~3).
The name `$q$-logarithm' appears to have been introduced by Umarov, Tsallis
and Steinberg in 2008~\cite{UTS}, working in statistical mechanics.

\item
\label{rmk:q-history-ent}
The $q$-logarithmic entropies have been discovered and rediscovered
repeatedly.  To my knowledge, they first appeared in a 1967 paper on
information and classification by Havrda and Charv\'at~\cite{HaCh}, who
used a form adapted to base~$2$ logarithms.  They were rediscovered in 1970
by Dar\'oczy~\cite{DaroGIF}.  The base~$e$ version $S_q$ seems to have
first appeared in a 1982 article of Patil and Taillie (\cite{PaTaDCM},
Section~3.2), where it was studied as an index of biodiversity.

In physics, meanwhile, the $q$-logarithmic entropies appeared in a 1971
article of Lindhard and Nielsen~\cite{LiNi} (according to
Csiszar~\cite{Csis}, Section~2.4), and in a 1978 survey by Wehrl
(\cite{Wehr}, p.~247).  Finally, they were rediscovered again in a 1988
paper on statistical physics by Tsallis~\cite{TsalPGB}.

Despite the twenty years of active life that the $q$-logarithmic entropies
had already enjoyed, it is after Tsallis that they are most commonly named.
The term `$q$-logarithmic entropy' is new, but has the benefits of being
descriptive and of not perpetuating a misattribution.

\item
\label{rmk:q-history-chn}
As in Remark~\ref{rmks:relent-condns}\bref{rmk:relent-condns-chn} or
Feinstein~\cite{Fein} (p.~5--6), a simple
inductive argument shows that the $q$-chain rule of
equation~\eqref{eq:q-chn} follows from the special case
\begin{equation}
\label{eq:q-chn-simp}
S_q\bigl(pw_1, (1 - p)w_1, w_2, \ldots, w_n\bigr)
=
S_q(\vc{w}) + w_1^q S_q(p, 1 - p)
\end{equation}
($p \in [0, 1]$, $n \geq 1$, $\vc{w} \in \Delta_n$).

\item
\label{rmk:q-history-char}
A characterization of the $q$-logarithmic entropies similar to
Theorem~\ref{thm:q} was published by Dar\'oczy in 1970~\cite{DaroGIF}.  He
assumed the full $q$-chain rule for $I(\vc{w} \of (\p^1, \ldots, \p^n))$
(in the equivalent form~\eqref{eq:q-chn-simp}), rather than just the
special case of $I(\vc{w} \otimes \p)$ that we used.  However, where we
assumed that $I \from \Delta_n \to \R$ is symmetric for all $n \geq 2$,
Dar\'oczy only assumed it for $n = 3$.  The two proofs are very different;
the main step in Dar\'oczy's was the solution of the functional
equation~\eqref{eq:feith-gen} in the case $f = g = h = k$.

Other characterizations of $S_q$ have been proved, using stronger
hypotheses than Theorem~\ref{thm:q} to obtain the same conclusion
(such as the theorem in Section~2 of~\cite{Suya}, and Theorem~V.2
of~\cite{Furu}).
\end{enumerate}
\end{remarks}

\section{$q$-logarithmic relative entropy}

For $q \neq 1$, the $q$-logarithmic relative entropy $S_q \from A_n \to
\R$, defined in the Introduction, is given explicitly by
\[
\srelent{q}{\p}{\vc{r}}
=
\frac{1}{q - 1} \Biggl( 
\sum_{i \csuch p_i > 0} p_i^q r_i^{1 - q} - 1 
\Biggr),
\]
for $(\p, \vc{r}) \in A_n$ and $q \neq 1$.  In the case $q = 1$, it reduces
to the ordinary relative entropy $\relent{\p}{\vc{r}}$.  As in that case,
restricting the arguments to lie in $A_n$ guarantees that
$\srelent{q}{-}{-}$ takes only finite values.

Our third and final theorem is a characterization of $q$-logarithmic
relative entropy, very similar to the characterization of $q$-logarithmic
entropy itself.

We begin by noting two properties of $q$-logarithmic relative entropy.
First, there is an easily-checked chain rule:
% 
% \begin{multline*}
\[
\Srelent{q}{\vc{w} \of (\p^1, \ldots, \p^n)}%
{\twid{\vc{w}} \of (\twid{\p}^1, \ldots, \twid{\p}^n)}
=
\srelent{q}{\vc{w}}{\twid{\vc{w}}}
+
\sum_{i \csuch w_i > 0} w_i^q \twid{w}_i^{1 - q}
\srelent{q}{\p^i}{\twid{\p}^i}
\]
% \end{multline*}
% 
($(\vc{w}, \twid{\vc{w}}) \in A_n$, $(\p^i, \twid{\p}^i) \in
A_{k_i}$).  This specializes to a $q$-multiplicativity formula
\begin{equation}
\label{eq:q-rel-mult}
\srelent{q}{\vc{w} \otimes \p}{\twid{\vc{w}} \otimes \twid{\p}}    
=
\srelent{q}{\vc{w}}{\twid{\vc{w}}}
+
\Biggl( \sum_{i \csuch w_i > 0} w_i^q \twid{w}_i^{1 - q} \Biggr)
\srelent{q}{\p}{\twid{\p}}
\end{equation}
($(\vc{w}, \twid{\vc{w}}) \in A_n$, $(\p, \twid{\p}) \in A_k$).  Second,
$q$-logarithmic relative entropy has the same symmetry property as ordinary
relative entropy:
\begin{equation}
\label{eq:q-rel-sym}
\srelent{q}{\p}{\vc{r}}
=
\srelent{q}{\p\sigma}{\vc{r}\sigma}
\end{equation}
for all $n \geq 1$, $(\p, \vc{r}) \in A_n$, and permutations $\sigma$ of
$\{1, \ldots, n\}$.  

% Much as in Theorem~\ref{thm:q}, we can exploit the symmetry-breaking in the
% $q$-multiplicativity formula to derive a simple characterization theorem.

\begin{thm}
\label{thm:q-rel}
Let $1 \neq q \in \R$ and let $\bigl( \melent{-}{-} \from A_n \to \R
\bigr)_{n \geq 1}$ be a sequence of functions.  The following are
equivalent:
\begin{enumerate}
\item
\label{part:q-rel-condns}
$\melent{-}{-}$ has the $q$-multiplicativity property~\eqref{eq:q-rel-mult}
and the symmetry property~\eqref{eq:q-rel-sym} (both with $I$ in place of
$S_q$);

\item
\label{part:q-rel-form}
$\melent{-}{-} = c\srelent{q}{-}{-}$ for some $c \in \R$.
\end{enumerate}
\end{thm}

\begin{proof}
It is trivial that~\bref{part:q-rel-form} implies~\bref{part:q-rel-condns}.
Now assume~\bref{part:q-rel-condns}.  By symmetry,
\[
\melent{\vc{w} \otimes \vc{p}}{\twid{\vc{w}} \otimes \twid{\vc{p}}}
=
\melent{\vc{p} \otimes \vc{w}}{\twid{\vc{p}} \otimes \twid{\vc{w}}}
\]
for all $n, k \geq 1$, $(\vc{w}, \twid{\vc{w}}) \in A_n$, and $(\vc{p},
\twid{\vc{p}}) \in A_k$.  So by $q$-multiplicativity,
\[
\melent{\vc{w}}{\twid{\vc{w}}}
+ 
\Biggl( \sum_{i \csuch w_i > 0} w_i^q \twid{w}_i^{1 - q} \Biggr)
\melent{\vc{p}}{\twid{\vc{p}}}
=
\melent{\vc{p}}{\twid{\vc{p}}}
+ 
\Biggl( \sum_{i \csuch p_i > 0} p_i^q \twid{p}_i^{1 - q} \Biggr)
\melent{\vc{w}}{\twid{\vc{w}}},
\]
or equivalently,
\[
\Biggl( \sum_{i \csuch w_i > 0} w_i^q \twid{w}_i^{1 - q} - 1 \Biggr) 
\melent{\vc{p}}{\twid{\vc{p}}}.
=
\Biggl( \sum_{i \csuch p_i > 0} p_i^q \twid{p}_i^{1 - q} - 1 \Biggr) 
\melent{\vc{w}}{\twid{\vc{w}}}.
\]
Take $\vc{w} = (1, 0)$ and $\twid{\vc{w}} = (1/2, 1/2)$: then
\[
(2^{q - 1} - 1) \melent{\vc{p}}{\twid{\vc{p}}}
=
\Biggl( \sum_{i \csuch p_i > 0} p_i^q \twid{p}_i^{1 - q} - 1 \Biggr)
\Melent{(1, 0)}{(1/2, 1/2)} 
\]
for all $(\p, \twid{\p}) \in A_k$.  But $q \neq 1$, so we can define 
\[
c = \frac{1 - q}{2^{q - 1} - 1} \cdot \Melent{(1, 0)}{(1/2, 1/2)},
\]
and then $\melent{-}{-} = c \srelent{q}{-}{-}$.
\end{proof}

\begin{remarks}
Other characterization theorems for $q$-logarithmic relative entropy have
been proved.  For example, Furuichi (\cite{Furu}, Section~IV) obtained the
same conclusion, but also assumed continuity and essentially
the full chain rule (that is, an equivalent special case, as in
Remarks~\ref{rmks:relent-condns}\bref{rmk:relent-condns-chn}
and~\ref{rmks:q-history}\bref{rmk:q-history-chn}).
%  Theorem~\ref{thm:q-rel} appears to be new.
\end{remarks}

\paragraph*{Acknowledgements} 
I thank John Baez and Tobias Fritz for their comments.  This work was
partially supported by a BBSRC FLIP award (BB/P004210/1).

\bibliography{mathrefs}

\end{document}